\documentclass[11pt]{article}

\usepackage{expcommon}

\newcommand{\Argmax}{\mbox{Argmax}}
\DeclareMathOperator{\EX}{\mathbb{E}}% expected value

 { \theoremstyle{plain} \newtheorem{assumption}{Assumption}
  }

\title{Density Estimation using Entropy Maximization for Semi-continuous Data}

\author[$1$]{Sai K. Popuri\thanks{For Correspondence: \url{saiku1@umbc.edu}. Much of the research in this paper was carried out when the first author was a graduate student at the Department of Mathematics and Statistics, University of Maryland, Baltimore County.}}
\author[$2$]{Nagaraj K. Neerchal}
\author[$3$]{Amita Mehta}
\author[$4$]{Ahmad Mousavi}
\scriptsize
\affil[$1$]{\footnotesize Walmart eCommerce, Sunnyvale, CA, USA}
\affil[$2$]{\footnotesize Department of Mathematics and Statistics, University of Maryland, Baltimore County, Baltimore, MD, USA}
\affil[$3$]{\footnotesize Joint Center for Earth Systems Technology, 5523 Research Park Dr., Baltimore, MD, USA}
\affil[$4$]{\footnotesize Institute for Mathematics and its Applications, University of Minnesota, Minneapolis, MN, USA }
\date{}

\begin{document}
\maketitle

\textsl{Revised version published in Digital Signal Processing (DOI: https://doi.org/10.1016/j.dsp.2021.103107)}

\begin{abstract}
Semi-continuous data comes from a distribution that is a mixture of the point mass at zero and a continuous distribution with support on the positive real line. A clear example is the daily rainfall data. In this paper, we present a novel algorithm to estimate the density function for semi-continuous data using the principle of maximum entropy. Unlike existing methods in the literature, our algorithm needs only the sample values of the constraint functions in the entropy maximization problem and does not need the entire sample. Using simulations, we show that the estimate of the entropy produced by our algorithm has significantly less bias compared to existing methods. An application to the daily rainfall data is provided.

\end{abstract}

\section{Introduction}
\label{sec:intro}

 Random variables that feature a point mass at zero and are continuous on the positive real line are sometimes known as semi-continuous 
 random variables\footnote{In the literature, the terms `semi-continuous' and `semicontinuous' seem to have been used interchangeably. In this paper, we use the term `semi-continuous.} and their distributions are known as semi-continuous distributions. A clear example of such random variables is the daily rainfall data. 
 Figure~\ref{fig:obshist} shows the histogram of the observed daily precipitation between years 1949 and 2000 at a location in the Midwest USA. Since it did not rain every single day at this location during this period, there are many zero values in the data. For the days it rained, a positive continuous number (e.g.: mm/day) is reported. The red dot on the histogram represents the proportion of zeros in the data. Suppose $Y$ is a semi-continuous random variable with a point mass of $\gamma$ at zero and has the density functrion $g_{\mathbf{\theta}}(y)$ on the positive 
 real line, where $\mathbf{\theta}$ is the parameter governing the distribution. Then, the probability density function (pdf) of $Y$ is given by (\citet{Popuri2017})
 \begin{equation}\label{eq:semicts-tp}
 f_\mathbf{\theta}(y) = \gamma\delta(y) + (1-\gamma)\delta^*(y) g_{\mathbf{\theta}}(y),
 \end{equation}
 where $\delta(y)$ is the indicator function of $y$ taking value $1$ if $y$ is $0$ and $0$ when $y$ is positive, and $\delta^*(y)=1-\delta(y)$.
 \begin{figure}[H]
 	\centering
 	\includegraphics[width=0.6\textwidth]{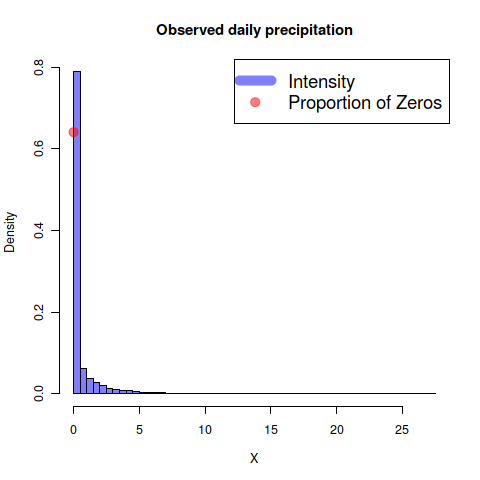}
 	\caption{Histogram of daily precipitation at a location}
 	\label{fig:obshist}
 \end{figure}
In the literature, the distribution in \eqref{eq:semicts-tp} is sometimes known as the `two-part' model owing to \citet{Cragg1971}, who has used it to analyse sales data for durable goods. The density function in \eqref{eq:semicts-tp} can be modified by restricting $\gamma$ to depend on $\mathbf{\theta}$, thereby inducing a dependency between the probability of observing a zero and the distribution of the continuous component, $g$. Let $\gamma$ be defined as a one-one function of $\mathbf{\theta}$  with the range in $(0,1)$. The pdf of $Y$ can now be written as
\begin{equation}\label{eq:semictsdensity}
f_\mathbf{\theta}(y) = \gamma(\mathbf{\theta})\delta(y) + (1-\gamma(\mathbf{\theta}))\delta^*(y) g_{\mathbf{\theta}}(y).
\end{equation}
When the density function $g$ in \eqref{eq:semictsdensity} is taken to be truncated normal with mean $\mu$ and variance $\sigma^2$, and $\gamma(\mathbf{\theta})$ is set equal to $1-\Phi(\frac{\mu}{\sigma})$, we get the density function in \eqref{eq:tobitdensity}, which is known as the Tobit model (\citet{Tobin1958}, \citet{Amemiya1985}), widely used in Econometrics,
\begin{equation}\label{eq:tobitdensity}
f_{\mathbf{\theta}}(y) = \Big(1-\Phi\Big(\frac{\mu}{\sigma}\Big)\Big)\delta(y)+\delta^*(y) \frac{1}{\sigma}\phi\Big(\frac{y-\mu}{\sigma}\Big), 
\end{equation}
where $\phi$ and $\Phi $ are the pdf and cumulative distribution function (cdf) of a standard normal random variable, respectively.

Estimation of probability density functions for semi-continuous data can be done using parametric, Bayesian, and non-parametric methods. Parametric methods consist of assuming a parametric distribution for the data and using methods like maximum likelihood estimation to estimate the unknown parameters. These methods provide efficient estimates but could be heavily model dependent with limited parametric forms. Estimation of the Tobit model in \eqref{eq:tobitdensity} is studied in detail in \citet{Amemiya1985}. Several authors have applied Bayesian methods to analyze semi-continuous data. For example, Gibbs sampling and the data augmentation method has been used for independent semi-continuous data in \citet{Chib1992}. \citet{Neelon2016} provides a detailed overview of various estimation methods, including Bayesian, applied for semi-continuous data in Biostatistics. Some of the non-parametric estimation methods include kernel density estimation, k-nearest neighbors etc. These methods are often computationally expensive and depend on choosing values for tuning parameters like bandwidth size. 

Estimation using entropy maximization is a non-parametric method that searches for a density function that maximizes the entropy value subject to certain constraints (\citet{Cover2006}). This approach is attractive in situations where we want the density function to satisfy certain properties but are otherwise non-commital with regard to other features. Estimates of density functions using this method are sometimes known as `MaxEnt' distributions. Depending on the choice of the constraint functions, this method could be computationally expensive. However, for certain standard choices of constraints, the MaxEnt distributions have closed forms. This could be useful, especially in the context of Bayesian estimation. 
%More recently, \citet{Popuri2017} has proposed an Expectation-Maximization (EM) like method for prediction of daily precipitation, wherein the zeros are imputed in the `E-step' and the Bayesian inference is made in the `M-step'.
In general, the posterior predictive distributions for semi-continuous data do not have closed forms and one needs a Markov Chain Monte Carlo (MCMC) like scheme to perform analysis (\citet{Popuri2017}). Even if they have closed forms, they are often not in recognizable classes of distributions and therefore not amenable to sampling in an automated programmatic fashion. As a result, such predictive distributions are often not suitable as inputs to systems that take distributions that are easy to sample from as input. For example, in hydrological studies, in order to assess the impact of changes in climate on water and crop yields, some researchers use a hydrological model called Soil and Water Assessment Tool (SWAT; \citet{SWAT2007}). Currently, SWAT takes point estimates of predictions of various climate variables, including daily precipitation, as inputs. However, it is conceivable that future versions of SWAT might take distributions of predictions as inputs instead. Such systems might not accept conditional distributions from a Bayesian fit and an MCMC scheme to sample predictions from posterior distributions as inputs. It would then be desirable to approximate the complicated predictive distributions to simpler forms, for example, using entropy maximization. 

Literature on entropy maximization for semi-continuous data is sparse. \citet{Politis1994} have presented a general form for the distribution estimated by entropy maximization for a mixture of discrete and continuous distributions when the MaxEnt distributions for the discrete and continuous parts are known. However, they do not provide a constructive algorithm to estimate density functions for mixture data, of which semi-continuous data is a special case. Density estimation by entropy maximization is widely used in hydrology. In hydrological studies, semi-continuous data is typically modelled by ignoring the point mass at zero and focusing exclusively on the positive continuous part (\citet{RHO2019210}). This is true for density estimation using entropy maximization as well. \citet{Papal2012} have applied various continuous distributions to model the positive continuous part of daily rainfall datasets but have set the unknown probability of observing a dry day (zero rainfall) to the proportion of zeros in the data before estimating density functions using entropy maximization for the positive continuous part alone. While this is convenient, the resulting MaxEnt density function is not the true MaxEnt solution and is often only a crude approximation, possibly with a substantial bias. In this paper we develop the theory of entropy maximization for semi-continuous data in a comprehensive way. Specifically, we do not set the point mass at zero to the proportion of zeros in the data. 
%In fact, our method does not need the entire data sample as an input. Instead, it needs only the values of the constraint functions used in the entropy maximization problem. 
We further present a novel algorithm that maximizes entropy using alternating estimation steps in an iterative fashion. We also extend the analysis in \citet{Politis1994} into an algorithm where the point mass at zero and the continuous parts are estimated in alternating steps. 
Rest of the paper is organized as follows.
%We find that approximation using the principle of maximum entropy to be a useful method for this purpose since the method is least prejudiced about the distribution beyond a set of constraints, which are provided based on the predictive density. 
Section~\ref{sec:method} provides a brief discussion of the method of entropy maximization, presents the theory suitably extended to semi-continuous data with examples, and provides an algorithm to 
estimate the density function for semi-continuous data. In section~\ref{sec:maxent-sim}, we evaluate our algorithm and compare it with two methods based on existing literature using simulations and a real-world dataset of daily rainfall at several locations in the Midwest US. We conclude with some comments on future work in section~\ref{sec:discussion}.

\section{Entropy maximization for semi-continuous data}
\label{sec:method}

The principle of maximum entropy is a method of estimation of probability distribution functions by maximizing the entropy of the distribution or the lack of information, subject to a set of constraints. Often, the constraints are on the moments of the distribution. The entropy (\citet{Cover2006}) of a distribution $p(y)$ with support on a subset $\mathcal{I}$ of the real line is defined as
\begin{definition}
	\label{def:entropy}
	$H(p) = -\int_{\mathcal{I}}p(y)\log p(y) \mu(dy)$.
\end{definition}
The density function $p$ is estimated by maximizing $H(p)$ subject to a set of constraints on functions of $Y$, often the moments of $Y$. Suppose we impose $K$ number of constraints given by
\begin{equation}
\label{eq:maxent-constraints}
\EX(h_j(y)) = \int_{\mathcal{I}}h_j(y)p(y) \mu(dy) = \alpha_j,
\end{equation}
where $\EX()$ is the Expectation operator, $j=0,\ldots,K$, $h_0(y)=1; \ \forall y \in \mathcal{I}$,  $\alpha_0=1$, and $\alpha_j \in (-\infty, \infty)$, $j=1,\ldots,K$. Note that the first constraint ensures that the desired density function integrates to $1$. The problem of estimating the density $p(y)$ using the principle of maximum entropy can be formulated as the following optimization problem:
\begin{equation}
\label{eq:maxent-ot}
\begin{aligned}
& \underset{p}{\text{maximize}}
& & H(p) \\
& \text{subject to}
& & \int_{\mathcal{I}}h_j(y)p(y) \mu(dy) = \alpha_j, \; j = 0, \ldots, K.
\end{aligned}
\end{equation}
The solution using the method of Lagrange multipliers is given by 
\begin{comment}
The Lagrangian for this problem is given by
\begin{equation}
\label{eq:maxent-lagrangian}
\mathcal{L} = -\int_{\mathcal{I}}p(y)\log p(y) \mu(dy) - \lambda_0 (\int_{\mathcal{I}}h_0(y)p(y) \mu(dy)-\alpha_0) - \ldots - \lambda_K (\int_{\mathcal{I}}h_K(y)p(y) \mu(dy)-\alpha_K).
\end{equation}
Then the solution has the necessary condition $\frac{\partial \mathcal{L}}{\partial p}=0$, from which we have
\begin{equation*}
\log p(y) = -(1 + \lambda_0 h_0(y)+\ldots+\lambda_K h_K(y))
\end{equation*}
and the resulting distribution is given by
\end{comment}

\begin{equation}
\label{eq:maxent-distr}
p(y) = \exp\Big\{-\big(1+\sum_{j=0}^{K}\lambda_j h_j(y)\big)\Big\}. 
\end{equation}
The probability density function in \eqref{eq:maxent-distr} is then substituted in the $K+1$ constraints in \eqref{eq:maxent-constraints} 
to solve for the Lagrange multipliers $\lambda_j$'s. The solution in \eqref{eq:maxent-distr} is sometimes known as 
the `MaxEnt' distribution. A well-known example of a MaxEnt solution is the normal distribution when the constraints are 
on the mean and variance for continuous distributions with support on the real line (\citet{Conrad2005}). Another example 
is the exponential distribution when the constraint is on the mean for continuous distributions with support on the positive real line (\citet{Conrad2005}). A third example is the gamma distribution when the constraints are on the mean and on the logarithm of the random variable for continuous distributions with support on the positive real line (\citet{Singh1985}). 

%%From the above examples, note that the MaxEnt distributions resulting from simple constraints for the continuous distributions are the well-known standard distributions. A natural question to ask is whether such simplifications are possible for semi-continuous distributions as well.
% and if so, how do the resulting MaxEnt distributions look like? 

Let $\mathscr{P}$ be the set of all density functions of semi-continuous random variables with support on $\mathcal{I}=[0,\infty)$, 
that is, continuous on the positive real line with a point mass at $0$. Recall that the density $p \in \mathscr{P}$ of a semi-continuous random variable $Y$ is given by 
\begin{equation}
\label{eq:maxent-semicts-density}
p(y) = \gamma\delta(y) + (1-\gamma)\delta^*(y)g(y),
\end{equation}
which is same as \eqref{eq:semicts-tp} but the parameter $\mathbf{\theta}$ is dropped for convenience in the notation.
It turns out that the entropy function for semi-continuous random variables with density function in \eqref{eq:maxent-semicts-density} is a function of the entropy of the constituent 
continuous part, as shown in the following result.
\begin{proposition} \label{prop:entropy-semicts}
	The entropy function of the semi-continuous distribution in \eqref{eq:maxent-semicts-density} is given by 
	\begin{equation}
	\label{eq:maxent-semicts-entropy}
	H(p) = -\gamma\log\gamma - (1-\gamma)\log(1-\gamma) + (1-\gamma)H(g),
	\end{equation}
	where $H(g)$ is the entropy function of $g$, the continuous component of the distribution.
\end{proposition}
\begin{proof}
	It is convenient to write the pdf in \eqref{eq:maxent-semicts-density} as
	\begin{equation}
	\label{eq:maxent-semicts-density2}
	p(y) = \gamma^{\delta(y)}(1-\gamma)^{\delta^*(y)}\big(g(y)\big)^{\delta^*(y)}.
	\end{equation}
	Taking $\log$ of the density function, we get 
	\begin{equation}
	 \label{eq:maxent-semicts-density3}
	 \log p(y) = \delta(y)\log \gamma+\delta^*(y)\log (1-\gamma)+\delta^*(y)\log\big(g(y)\big).
 	\end{equation}
 	Substituting the density function in \eqref{eq:maxent-semicts-density2} and its $\log$ in \eqref{eq:maxent-semicts-density3} in the definition of entropy in \eqref{def:entropy} and noting that the dominating measure for semi-continuous random variables is the sum of the counting and Lebesgue measures, and that their support is $\mathcal{I} = \{0\}\cup (0,\infty)$, we get 
	\begin{equation}
	\label{eq:entropy-semicts}
	\begin{aligned}
	H(p) &= -\int_{\mathcal{I}}p(y)\log p(y) \mu(dy) \\
	&=
	\begin{aligned}[t]
	& -\Big(\int_{\mathcal{I}}\delta(y)\gamma^{\delta(y)}(1-\gamma)^{\delta^*(y)}\big(g(y)\big)^{\delta^*(y)}\log\gamma \mu(dy) \\ 
	& + \int_{\mathcal{I}}\delta^*(y)\gamma^{\delta(y)}(1-\gamma)^{\delta^*(y)}\big(g(y)\big)^{\delta^*(y)}\log (1-\gamma) \mu(dy) \\
	& + \int_{\mathcal{I}}\delta^*(y)\gamma^{\delta(y)}(1-\gamma)^{\delta^*(y)}\big(g(y)\big)^{\delta^*(y)}\log g_{\mathbf{\theta}}(y) \mu(dy)\Big)
	\end{aligned} \\
	&= -\Big(\gamma\log\gamma + (1-\gamma)\log (1-\gamma) + \int_0^{\infty}(1-\gamma)g(y)\log g_{\mathbf{\theta}}(y)dy\Big) \\
	&= -\gamma\log\gamma - (1-\gamma)\log (1-\gamma) + (1-\gamma)H(g).
	\end{aligned}
	\end{equation}
\end{proof}
Density function estimation using entropy maximization for semi-continuous data involves maximizing the entropy function in \eqref{eq:entropy-semicts}, similar to the problem in \eqref{eq:maxent-ot}. Consider the following optimization problem: maximize the entropy $H(p)$ in \eqref{eq:entropy-semicts} over all probability density functions $p\in \mathscr{P}$ satisfying
\begin{equation} 
\label{eq:semicts-maxent-constr}
\int_{\mathcal{I}}h_j(y)p(y)\mu(dy) =  \alpha_j, j=0,\ldots,K,
\end{equation}
where $h_0(y) = 1; \ \forall y \in \mathcal{I}$, $\alpha_0 = 1$, and 
$\alpha_j \in (0, \infty)$, $j=1,\ldots,K$. The corresponding Lagrange function is given by 
\begin{equation}
\label{eq:semictsmaxent-lagrangian}
\begin{aligned}
\mathcal{L} &= H(p) - \lambda_0 \Big(\int_{\mathcal{I}}h_0(y)p(y) \mu(dy)-\alpha_0\Big) - \ldots - \lambda_K \Big(\int_{\mathcal{I}}h_K(y)p(y) \mu(dy)-\alpha_K\Big) \\
&= -\gamma\log\gamma - (1-\gamma)\log (1-\gamma) + (1-\gamma)H(g) - \sum_{j=0}^{K}\lambda_j\Big\{\gamma h_j(0) + (1-\gamma)\int_0^{\infty}h_j(y)g(y)dy - \alpha_j\Big\}.
\end{aligned}
\end{equation}
Setting the derivatives of $\mathcal{L}$ with respect to $\gamma$ and $g$ to $0$ and solving the Lagrange equations, we obtain the following maximizing density function
\begin{equation}
\label{eq:maxent-semicts-density1}
p(y) = \gamma^*\delta(y) + (1-\gamma^*)\delta^*(y)g^*(y),
\end{equation}
where
\begin{equation}
    \label{eq:semicts-maxent-soln-g}
    g^*(y) = \exp\Big\{-\big(1+\sum_{j=0}^{K}\lambda_j h_j(y)\big)\Big\},
\end{equation}
and 
\begin{equation}
    \label{eq:semicts-maxent-soln-gamma}
    \gamma^* = \frac{1}{1+\text{exp}\{H(g^*)+\sum_{j=0}^{K}\lambda_j h_j(0)-\sum_{j=0}^{K}\lambda_j\int_{0}^{\infty}h_j(y)g^*(y)dy\}}.
\end{equation}
The Lagrange multipliers $\lambda$'s are obtained by substituting the density in \eqref{eq:maxent-semicts-density1} in the constraints in \eqref{eq:semicts-maxent-constr} and solving for $\lambda$'s.

Some of the results for standard choices of constraints for continuous data extend to the semi-continuous case.
\begin{example}\label{ex:maxent-exp}(Semi-continuous Exponential distribution)\newline
    We know that 
    \begin{equation}
	\label{eq:exp-ex}
	g^*(y) = \frac{1}{\alpha_1} \text{exp}\Big(-\frac{y}{\alpha_1}\Big)
	\end{equation}
    is the distribution that maximizes the entropy in the class of density functions for continuous random variables with support on the positive real line and the constraint on the first moment with $\int_{0}^{\infty}yg(y)dy = \alpha_1$. In particular, suppose the data generating process (DGP) is assumed to be the exponential distribution $g(y) = \theta\text{exp}(-\theta y)$ and $\alpha_1$ is set to the expected value $1/\theta$, then the MaxEnt distribution is $g(y)$, the DGP that we started with.
    
    Extending to the semi-continuous case, we will see that when the constraint is on the first moment, the MaxEnt distribution is a semi-continuous density function where the continuous component is the exponential distribution. However, the MaxEnt distribution when the DGP is a two-part exponential distribution (\eqref{eq:semicts-tp}, where $g_{\theta}(y)$ is exponential), is not the DGP, unlike the continuous case.
    
    The constraint on the first moment is given by 
	\begin{equation}
	\label{eq:exp-example-constr}
	\int_{\mathcal{I}}h_1(y)p(y)\mu(dy) = \alpha_1,
	\end{equation} 
	where $h_1(y) = y; \  \forall y \in \mathcal{I}$.
	Then the MaxEnt distribution $p$ in \eqref{eq:maxent-semicts-density1} is given by 
	\begin{equation}
	\label{eq:exp-example1}
		p(y) = \gamma^*\delta(y) + (1-\gamma^*)\delta^*(y)g^*(y),
	\end{equation}
	where 
	\begin{equation*}
		g^*(y) = \text{exp}(-1-\lambda_0-\lambda_1 y)
	\end{equation*}
	and
	\begin{equation}
	\label{eq:exp-example2}
		\gamma^* = \frac{1}{1+\text{exp}\big\{H(g^*) - \lambda_1\int_{0}^{\infty}yg^*(y)dy\big\}}.
	\end{equation}
	Substituting $p(y)$ from \eqref{eq:exp-example1} into the first constraint $\int_{\mathcal{I}}p(y)\mu(dy) = 1$ and the second constraint in \eqref{eq:exp-example-constr}, and solving for $\lambda_0$ and $\lambda_1$, we get
	\begin{equation}
	\label{eq:exp-example3}
		g^*(y) = \frac{1-\gamma^*}{\alpha_1}\text{exp}\Big(-\frac{1-\gamma^*}{\alpha_1}y\Big).
	\end{equation}
	Substituting $g^*$ from \eqref{eq:exp-example3} in \eqref{eq:exp-example2} simplifies to the quadratic equation in $x$
	\begin{equation}
	\label{eq:exp-example4}
		x^2 - x(2+\alpha_1) + 1 = 0,
	\end{equation}
whose solution is $\gamma^*$. Since $\gamma^* \in (0,1)$, the only feasible solution is $\frac{(2+\alpha_1) - \sqrt{(2+\alpha_1)^2 - 4}}{2}$. 
Suppose the DGP is assumed to be the following two-part exponential distribution
\begin{equation}
\label{eq:exp-example-dgp}
    p(y) = \gamma\delta(y) + (1-\gamma)\delta^*(y)\theta\text{exp}(-\theta y),
\end{equation}
where $\gamma \in (0, 1)$ and $\theta > 0$ are known. Further suppose we set $\alpha_1$ to $E(Y) = (1-\gamma)\frac{1}{\theta}$. By the analogy in the case of the continuous data with exponential distribution as the DGP, one would expect the MaxEnt distribution to be \eqref{eq:exp-example-dgp}. However, the MaxEnt distribution is given by 
\begin{equation*}
    p(y) = \gamma^*\delta(y) + (1-\gamma^*)\delta^*(y)g^*(y),
\end{equation*}
where 
\begin{equation*}
	g^*(y) = \frac{\theta(1-\gamma^*)}{(1-\gamma)}\text{exp}\Big(\frac{-y\theta(1-\gamma^*)}{(1-\gamma)}\Big).
\end{equation*}
and
$\gamma^*$ is the solution to the quadratic equation in variable $x$
\begin{equation*}
	\theta x^2 - x(1-\gamma+2\theta) + \theta = 0,
\end{equation*}
where the only feasible solution is $\frac{(1-\gamma+2\theta)-\sqrt{(1-\gamma+2\theta)^2 - 4\theta^2}}{2\theta}$. 
Notice that even though the MaxEnt solution is not the DGP two-part exponential distribution, it still belongs to the class of two-part exponential distributions.
\end{example}
	
\begin{example}\label{ex:maxent-gamma}(Semi-continuous Gamma distribution)\newline
    The two parameter Gamma distribution maximizes the entropy in the class of density functions with support on the positive real line and constraints on the first moment and the logarithm of the variable (\citet{Singh1985}). Suppose the DGP is given by the following gamma distribution
    \begin{equation}
        \label{eq:exp-gamma-dgp}
        g(y) = \frac{1}{\Gamma(\kappa)\theta^\kappa}y^{\kappa-1}e^{-y/\theta},
    \end{equation}
    where $\kappa > 0$ and $\theta > 0$ are the shape and scale parameters, respectively. Further suppose that in the entropy maximization problem,
    \begin{align*}
        \begin{split}
             h_1(y) &= y ,
            \\
             \alpha_1 &= \kappa\theta = \EX(Y)
        \end{split}
    \end{align*}
    and 
    \begin{align*}
        \begin{split}
             h_2(y) &= \log(y) ,
            \\
             \alpha_2 &= \log(\theta) + \psi(\kappa) = \EX(\log(Y)).
        \end{split}
    \end{align*}
    Then the MaxEnt distribution is the DGP in \eqref{eq:exp-gamma-dgp}.
    
    Consider the semi-continuous case with support on $[0, \infty]$ and suppose we choose the following $h_j$ functions
    \begin{align*}
        \begin{split}
             h_1(y) &= y ,
            \\
             h_2(y) &= \begin{cases}
                0 &\text{$y = 0$}\\
                \log(y) &\text{$y > 0$}.
                \end{cases}
        \end{split}
    \end{align*}
    The corresponding $\alpha_1$ and $\alpha_2$ are typically empirical averages of the above functions based on a sample and therefore, are known. 
    Then the MaxEnt distribution $p$ in \eqref{eq:maxent-semicts-density1} is given by 
	\begin{equation}
	\label{eq:gamma-example1}
		p(y) = \gamma^*\delta(y) + (1-\gamma^*)\delta^*(y)g^*(y),
	\end{equation}
	where 
	\begin{equation*}
	\label{eq:gamma-example2}
		g^*(y) = \text{exp}(-1-\lambda_0-\lambda_1 y - \lambda_2 I(y>0)\log(y))
	\end{equation*}
	and
	\begin{equation}
	\label{eq:gamma-example3}
		\gamma^* = \frac{1}{1+\text{exp}\big\{H(g^*) - \lambda_1\int_{0}^{\infty}yg^*(y)dy - \lambda_2\int_{0}^{\infty}\log(y)g^*(y)dy\big\}}.
	\end{equation}
	Following the algebra in \citet{Singh1985}, we substitute the density function in \eqref{eq:gamma-example1} in the constraints and solve for the Lagrange multipliers.
	
	First, we equate the integral of the density function in \eqref{eq:gamma-example1} over the support of the semi-continuous random variable to $1$ to get
	\begin{flalign}
	\begin{aligned}
	\label{eq:ex-gamma-constr1}
        & \int_{\mathcal{I}}p(y)\mu(dy) =  1 \\ 
        \implies & \gamma^* + (1-\gamma^*)\int_{0}^{\infty}\text{exp}\big\{-1-\lambda_0-\lambda_1y-\lambda_2\log(y)\big\}dy = 1 \\
        \implies & \int_{0}^{\infty}\exp\{-\lambda_1y-\lambda_2\log(y)\}dy = \exp\{1+\lambda_0\} \\
        \implies & \lambda_0 = \log(\Gamma(1-\lambda_2)) - (1-\lambda_2)\log(\lambda_1)-1.
    \end{aligned}
    \end{flalign}
    The second constraint yields
    \begin{flalign}
	\begin{aligned}
	\label{eq:ex-gamma-constr2}
        & \int_{\mathcal{I}}yp(y)\mu(dy) =  \alpha_1 \\ 
        \implies & (1-\gamma^*)\int_{0}^{\infty}y\text{exp}\big\{-1-\lambda_0-\lambda_1y-\lambda_2\log(y)\big\}dy = \alpha_1 \\
        \implies & \int_{0}^{\infty}y\exp\{-\lambda_1y-\lambda_2\log(y)\}dy = \frac{\alpha_1\exp\{1+\lambda_0\}}{1-\gamma^*},
    \end{aligned}
    \end{flalign}
    and the third constraint gives
    \begin{flalign}
	\begin{aligned}
	\label{eq:ex-gamma-constr3}
        & \int_{\mathcal{I}}I(y>0)\log(y)p(y)\mu(dy) =  \alpha_2 \\
        \implies & \int_{0}^{\infty}\log(y)\exp\{-\lambda_1y-\lambda_2\log(y)\}dy = \frac{\alpha_2\exp\{1+\lambda_0\}}{1-\gamma^*}.
    \end{aligned}
    \end{flalign}
    Differentiating $\exp\{1+\lambda_0\}$ in \eqref{eq:ex-gamma-constr1} with respect to $\lambda_1$, we get
    \begin{flalign}
	\begin{aligned}
	\label{eq:ex-gamma-eq1}
        & 
        \begin{split}
            \exp\{1+\lambda_0\}\frac{\partial\lambda_0}{\partial\lambda_1} & = -\int_{0}^{\infty}y\exp\{-\lambda_1y-\lambda_2\log(y)\}dy \\
            & = -\alpha_1\frac{\exp\{1+\lambda_0\}}{1-\gamma^*} \text{(from \eqref{eq:ex-gamma-constr2})}
        \end{split}
        \\
        \implies & \frac{\partial\lambda_0}{\partial\lambda_1} = -\frac{\alpha_1}{1-\gamma^*}.
    \end{aligned}
    \end{flalign}
    Differentiating $\lambda_0$ in \eqref{eq:ex-gamma-constr1} with respect to $\lambda_1$, we get
    \begin{equation}
        \label{eq:ex-gamma-eq2}
        \frac{\partial\lambda_0}{\partial\lambda_1} = -\frac{1-\lambda_2}{\lambda_1}.
    \end{equation}
    From \eqref{eq:ex-gamma-eq1} and \eqref{eq:ex-gamma-eq2}, we get
    \begin{equation}
        \label{eq:ex-gamma-eq3}
        \frac{1-\lambda_2}{\lambda_1} = \frac{\alpha_1}{1-\gamma^*}.
    \end{equation}
    Differentiating $\exp\{1+\lambda_0\}$ in \eqref{eq:ex-gamma-constr1} with respect to $\lambda_2$, we get
    \begin{flalign}
	\begin{aligned}
	\label{eq:ex-gamma-eq4}
        & 
        \begin{split}
            \exp\{1+\lambda_0\}\frac{\partial\lambda_0}{\partial\lambda_2} & = -\int_{0}^{\infty}\log(y)\exp\{-\lambda_1y-\lambda_2\log(y)\}dy \\
            & = -\alpha_2\frac{\exp\{1+\lambda_0\}}{1-\gamma^*} \text{(from \eqref{eq:ex-gamma-constr3})}
        \end{split}
        \\
        \implies & \frac{\partial\lambda_0}{\partial\lambda_2} = -\frac{\alpha_2}{1-\gamma^*}.
    \end{aligned}
    \end{flalign}
    Differentiating $\lambda_0$ in \eqref{eq:ex-gamma-constr1} with respect to $\lambda_2$, we get
    \begin{equation}
        \label{eq:ex-gamma-eq5}
        \frac{\partial\lambda_0}{\partial\lambda_2} = \log(\lambda_1) - \psi(1-\lambda_2).
    \end{equation}
    From \eqref{eq:ex-gamma-eq4} and \eqref{eq:ex-gamma-eq5}, we get
    \begin{equation}
        \label{eq:ex-gamma-eq6}
        \psi(1-\lambda_2) - \log(\lambda_1) = \frac{\alpha_2}{1-\gamma^*}.
    \end{equation}
    Substituting $\lambda_0$ from \eqref{eq:ex-gamma-constr1} in \eqref{eq:gamma-example1}, we get 
    \begin{equation} \label{eq:gamma-soln}
        \begin{split}
            p(y) & = \gamma^*\delta(y) + (1-\gamma^*)\delta^*(y)\exp\{-\log(\Gamma(1-\lambda_2))+(1-\lambda_2)\log(\lambda_1)-\lambda_1y-\lambda_2I(y>0)\log(y)\} \\
                & = \gamma^*\delta(y) + (1-\gamma^*)\delta^*(y)\frac{\lambda_1^{1-\lambda_2}}{\Gamma(1-\lambda_2)}\exp\{-\lambda_1y\}y^{-\lambda_2I(y>0)} \\
                & = \gamma^*\delta(y) + (1-\gamma^*)\delta^*(y)\frac{1}{\Gamma(a)b^a}\exp\{-y/b\}y^{a-1},
        \end{split}
    \end{equation}
    where $a = 1-\lambda_2$, $b=1/\lambda_1$, and 
    \begin{equation}
        \label{eq:gamma-gamma-soln}
        \begin{split}
        \gamma^* & = \frac{1}{1+\exp\big\{H(g^*)-1+\lambda_2+\lambda_2\log(\lambda_1)-\lambda_2\psi(1-\lambda_2)\big\}} \\
            & = \frac{1}{1+\exp\big\{1-\lambda_2-\log(\lambda_1)+\log\Gamma(1-\lambda_2)+\lambda_2\psi(1-\lambda_2)-1+\lambda_2+\lambda_2\log(\lambda_1)-\lambda_2\psi(1-\lambda_2)\big\}} \\
            & = \frac{1}{1+\exp\{-(1-\lambda_2)\log(\lambda_1) + \log\Gamma(1-\lambda_2)\}} \\ 
            & = \frac{\lambda_1^{1-\lambda_2}}{\lambda_1^{1-\lambda_2}+\Gamma(1-\lambda_2)}.
        \end{split}
    \end{equation}
    The unknown Lagrange multipliers are found by solving \eqref{eq:ex-gamma-eq3}, \eqref{eq:ex-gamma-eq6}, and \eqref{eq:gamma-gamma-soln}. Notice that unlike the semi-continuous exponential distribution in Example~\ref{ex:maxent-exp}, there is no closed form solution here. Similar to the semi-continuous exponential distribution in Example \ref{ex:maxent-exp}, suppose DGP is a semi-continuous gamma distribution with known parameter values and suppose $\alpha_1$ and $\alpha_2$ are set to the respective true expected values. Because the closed form solution does not exist here, the DGP cannot be recovered as the MaxEnt solution. However, notice that the MaxEnt distribution in \eqref{eq:gamma-soln} belongs to the class of semi-continuous gamma distributions.
\end{example}
\begin{comment}
\begin{remark}
	While the entropy maximization method yields standard distributions for certain constraints as we saw in Corollaries~\ref{ex:maxent-exp} and \ref{ex:maxent-gamma}, such simplifications are not always possible for more complicated constraints. As a result, if it is desired to estimate the complicated posterior predictive distributions with standard distributions that allow sampling using off-the-shelf software, then we might be restricted on the choice of constraints. 
\end{remark}
\end{comment}

\subsection{An alternating entropy maximization algorithm for semi-continuous data}
\label{sec:algo}

Suppose that the functions $h_j(y)$'s are known and differentiable on $S$ where $[0,\infty)\subseteq S$, and scalars $\alpha_j$, for $j=1,\dots, K,$ are given.  Further, assume that there exists  $p\in \mathscr{P}$ such that it satisfies the following constraints:
\begin{equation} \label{eq:thm1-constr}
	\int_{\mathcal{I}}h_j(y)p(y)\mu(dy) = \alpha_j, \quad \text{ for }\quad  j=0,\ldots,K \quad \text{ with } \quad h_0 = \alpha_0 = 1.
	\end{equation}
Then, our goal is to identify the properties of the distribution maximizing the entropy of a semi-continuous random variable with density in $\mathscr{P}$ subject to the above constraints. \begin{comment}
\begin{equation*}
	p^*(y) = \gamma^*\delta(y) + (1-\gamma^*)(1-\delta(y))g^*(y),
	\end{equation*}
where $g^*$ maximizes the entropy of a continuous random variable with support $(0,\infty)$ subject to the constraints:
	\begin{equation} 	\label{eq:thm1-constr-g^*}
	 \int_{0}^{\infty}h_j(z)g^*(z)dz = \frac{\mu_j-\gamma^*h_j(0)}{1-\gamma^*}, \quad  j=0,\ldots,K
	\end{equation}
and $\gamma^* = \frac{1}{1+\exp(H(g^*))}$ in which $H(g^*)$ is the maximized entropy value of $g^*$.
\end{comment}
The constraints in \eqref{eq:thm1-constr} are expanded as 
\begin{equation*}
	\int_{\mathcal{I}}h_j(y)p(y)\mu(dy) = \gamma h_j(0) + (1-\gamma)\int_{0}^{\infty}h_j(y)g(y)dy = \alpha_j,
\end{equation*}
which give us the following constraints:
\begin{equation} \label{eq:thm1-constr1}
	\int_{0}^{\infty}h_j(y)g(y)dy = \frac{\alpha_j - \gamma h_j(0)}{1-\gamma}, \quad 	j=0,\ldots,K.
\end{equation}
In particular, the function $g^*$ must satisfy 
	\begin{equation} 	\label{eq:thm1-constr-g^*}
	 \int_{0}^{\infty}h_j(z)g^*(z)dz = \frac{\alpha_j-\gamma^*h_j(0)}{1-\gamma^*}, \quad  j=0,\ldots,K
	\end{equation}
Note that for the  first constraint, the quantity $\int_{0}^{\infty}h_0(y)g(y)dy$ evaluates to $1$ since $\alpha_0 = h_0 = 1$.
\newline
For capturing the properties of $\gamma^*$ and  $g^*$, we bring three definitions as follows: 
\begin{equation*}
H(\gamma, g):=-\gamma\log\gamma - (1-\gamma)\log (1-\gamma) + (1-\gamma)H(g),
\end{equation*}
so that  $H(p)=H(\gamma,g)$ based on (\ref{eq:entropy-semicts}), and 
\begin{equation} \label{def:g-mid-gamma}
\left\{ \begin{array}{ll}
\mathscr{G}:=\{g \,|\, g \text{ is the  density function of a continuous random variable supported on } (0,\infty)\}; \\
\mathscr{G}\mid \gamma:=\{g \in \mathscr{G}\,|\, g \text{ satisfies }  (\ref{eq:thm1-constr1}) \text{ for } \gamma\}.\end{array} \right.
\end{equation}
We now divide the optimization problem into two interdependent sub-problems, as given in the following result. This decomposition will allow our algorithm to estimate the density function in an alternating iterative procedure.
\begin{proposition}\label{thm:maxent-obj-decomposition}
The problem of maximization of $H(p)$ in \eqref{eq:maxent-semicts-entropy} can be solved using the following equivalent problem:
\begin{equation*}
    \max_{p\in \mathscr{P}} \, \{H(p) \,|\, p \text{ satisfies } (\ref{eq:thm1-constr})\} = \max_{\gamma\in (0,1)}\big\{-\gamma\log\gamma - (1-\gamma)\log(1-\gamma)+(1-\gamma)\max_{g\in \mathscr{G}\mid \gamma} H(g)\big\}.
\end{equation*}
\end{proposition}
\begin{proof}
We clearly have
\begin{equation*} 
\max_{p\in \mathscr{P}} \, \{H(p) \,|\, p \text{ satisfies } (\ref{eq:thm1-constr})\}  = \max_{\gamma\in (0,1), \, g\in \mathscr{G}} \, \{H(\gamma,g)  \,|\,  \gamma, g \text{ satisfy }  (\ref{eq:thm1-constr1})\}.
\end{equation*}
Next, we demonstrate that
\begin{equation} \label{eq:eqaulity-max-max,max}
\max_{\gamma\in (0,1), \, g\in \mathscr{G}} \, \{H(\gamma,g)  \,|\,  \gamma, g \text{ satisfy }  (\ref{eq:thm1-constr1})\} = \max_{\gamma\in (0,1)}\max_{g\in \mathscr{G}\mid \gamma} \ H(\gamma,g) .
\end{equation}
By defining
\begin{equation} \label{def:little-h-gamma}
h(\gamma):=\left\{ \begin{array}{ll}
           \max_{g\in \mathscr{G}\mid \gamma} H(\gamma, g) & \mbox{if $\exists g \text{ that satisfies }  (\ref{eq:thm1-constr1}) \text{ for } \gamma $};\\
        -\infty  & \mbox{if $\not \exists g \text{ that  satisfies }  (\ref{eq:thm1-constr1}) \text{ for } \gamma$},\end{array} \right.
\end{equation}
the equality  (\ref{eq:eqaulity-max-max,max}) is equivalent to 
\begin{equation*} %\label{eq:eqaulity-max-max,max-hgamma}
\max_{\gamma\in (0,1), \, g\in \mathscr{G}} \, \{H(\gamma,g)  \,|\,  \gamma, g \text{ satisfy }  (\ref{eq:thm1-constr1})\} = \max_{\gamma\in (0,1)} \, h(\gamma).
\end{equation*}
To prove the equality above, we first fix some $ \gamma \in (0,1)$ and let $ g$ be a maximizer of $h( \gamma)$.  Then, let  $\bar \gamma\in \Argmax \, h( \gamma)$ over $ \gamma \in (0,1)$ and $\bar g\in \Argmax \, H(\bar \gamma,g)$ such that  $g$ satisfies  (\ref{eq:thm1-constr1}) for $\bar \gamma$.  Notice that $h(\bar \gamma)\ne -\infty,$ by the assumption on the existence of $p\in\mathscr P$ satisfying (\ref{eq:thm1-constr}). Clearly,  the  pair $(\bar \gamma, \bar g)$ satisfies (\ref{eq:thm1-constr1}) and $ \max_{\gamma\in (0,1), \, g\in \mathscr{G}\mid \gamma} \, \{H(\gamma,g)  \,|\,  \gamma, g \text{ satisfy }  (\ref{eq:thm1-constr1})\}\ge H(\bar \gamma,\bar g)$. This shows that $\max_{\gamma\in (0,1), \, g\in \mathscr{G}\mid \gamma} \, \{H(\gamma,g)  \,|\,  \gamma, g \text{ satisfy }  (\ref{eq:thm1-constr1})\}\ge \max_{\gamma\in (0,1)} \, h(\gamma).$
Conversely, suppose that $\gamma \in (0,1)$ and $(\gamma,g)$ satisfies the constraints (\ref{eq:thm1-constr1}), then by definition, we have $h(\gamma)\ge H(\gamma, g)$. Hence, $\max_{\gamma\in (0,1), \, g\in \mathscr{G}\mid \gamma} \, \{H(\gamma,g)  \,|\,  \gamma, g \text{ satisfy }  (\ref{eq:thm1-constr1})\}\le \max_{\gamma\in (0,1)} \, h(\gamma).$

By taking advantage of the proved equality (\ref{eq:eqaulity-max-max,max}), the problem of maximization of the entropy in \eqref{eq:maxent-semicts-entropy} can be broken down as 
	\begin{equation}
	\label{eq:thm1-maxent}
	\begin{aligned}
 \max_{p\in \mathscr{P}} \, \{H(p) \,|\, p \text{ satisfies } (\ref{eq:thm1-constr})\}  &= 
 \max_{\gamma\in (0,1), \, g\in \mathscr{G}} \, \{H(\gamma,g)  \,|\,  \gamma, g \text{ satisfy }  (\ref{eq:thm1-constr1})\} 
 	\\&= 
\max_{\gamma\in (0,1)}\max_{g\in \mathscr{G}\mid \gamma} \, H(\gamma,g)  
	\\&= 
\max_{\gamma\in (0,1)}\big\{-\gamma\log\gamma - (1-\gamma)\log(1-\gamma)+(1-\gamma)\max_{g\in \mathscr{G}\mid \gamma} H(g)\big\}.
\end{aligned}
\end{equation}
\end{proof}

\begin{comment}
\begin{assumption}
For each $\gamma\in (0,1)$, the corresponding  set $\mathscr{G}\mid \gamma$ defined in (\ref{def:g-mid-gamma}) is nonempty.
\end{assumption}
\end{comment}

Under the assumption that for each $\gamma\in (0,1)$, the corresponding  set $\mathscr{G}\mid \gamma$ defined in (\ref{def:g-mid-gamma}) is nonempty, for each $\gamma \in (0,1)$, the set $\mathscr{G}\mid \gamma$ is a  nonempty convex feasible set so that the concave maximization problem $\max_{g\in \mathscr{G}\mid \gamma} H(g)$ has a unique solution due to  \citet{Cover2006}; in the sense that the optimums do not differ except from a measure zero set.  Theorem 12.1.1 in \citet{Cover2006} establishes that the unique   optimum $g \mid \gamma= \Argmax_{g \in \mathscr{G} \mid \gamma}H(g)=\exp\Big(\lambda_0(\gamma)-1+ \sum_{j=1}^K \lambda_j(\gamma) h_j(y)\Big)$ such that $\lambda_j(\gamma)$'s are chosen so that $g\mid \gamma \in \mathscr{G} \mid \gamma$. Thus, if each $\lambda_j(\gamma)$ for $j=1,\dots, K$ is differentiable, the optimum  $g\mid \gamma$ is differentiable with respect to $\gamma$, and consequently, the concave  $h(g\mid \gamma)=-\int (g\mid \gamma)\log (g\mid \gamma)$,  for each $\gamma \in (0,1)$, defined in is well-defined and differentiable with respect to $\gamma$. 
\newline
Thus, under such assumptions and by Proposition~\ref{thm:maxent-obj-decomposition}, we have
\begin{equation*}
 \max_{p\in \mathscr{P}} \, \{H(p) \,|\, p \text{ satisfies } (\ref{eq:thm1-constr})\} =\max_{\gamma\in (0,1)}\big\{-\gamma\log\gamma - (1-\gamma)\log(1-\gamma)+(1-\gamma)h(\gamma)\big\}.
\end{equation*}
This further implies that $\gamma^*$ optimizing the right-hand side of an equality above must satisfy the 
\textit{implicit} equation as follows:
\begin{equation}\label{eq:gamma-star}
\log\Big(\frac{1-\gamma}{\gamma}\Big) -h(\gamma) -(\gamma  1)h'(\gamma)=0.
\end{equation}
Since the closed-form representations of $h(\gamma)$ and its derivative are not available in practice, motivated by the equation (\ref{eq:gamma-star}), we propose a two-step iterative procedure.  The Algorithm~\ref{alg:maxent-semicts-algo} aims to determine the MaxEnt distribution for semi-continuous data alternatively in which
\begin{equation}
\label{eq:backward-diff}
  h'(\gamma^{(k)})=\frac{h(\gamma_k)-h(\gamma_{k-1})}{\gamma_k-\gamma_{k-1}},  
\end{equation}
based on the backward difference  idea for numerically approximating derivatives. Thus, once $h(\gamma_k)$ is in hand, we numerically solve the following equation to update the $\gamma_{k+1}$: 
\begin{equation} \label{eq:thm1-maxgamma}
\log\Big(\frac{1-\gamma_{k+1}}{\gamma_{k+1}}\Big) -h(\gamma_k) -(\gamma_{k+1}-1)\frac{h(\gamma_k)-h(\gamma_{k-1})}{\gamma_k-\gamma_{k-1}}=0.
\end{equation}
\begin{algorithm}[H]
	\caption{Alternating Entropy Maximization for Semi-continuous Data (AEM)}
	\label{alg:maxent-semicts-algo}
	\begin{algorithmic}
		\State Step 1a: Choose small positive real values as tolerance error $\epsilon$, and a maximum number of iterations $M$.
		\State Step 1c: Choose initial values $\gamma^{(0)}$ and $\gamma^{(-1)}$ for $\gamma$ and calculate $H(g^{*(0)})$ and $H(g^{*(-1)})$, respectively, where $g^{*}$ is the MaxEnt distribution with constraints in \eqref{eq:thm1-constr1} with $\gamma^{(0)}$ (or $\gamma^{(-1)}$) substituted for $\gamma$. Set $e$ to $H(p^{*(0)})$. Set $k=1$.
		\State Step 2: 
		\While{$k \leq M$ OR $e > \epsilon$}.
		\State Calculate $h'(\gamma^{(k)})$ using \eqref{eq:backward-diff}.
		\State Set $\gamma^{(k)}$ using the equality in \eqref{eq:thm1-maxgamma} and compute $H(g^{*(k)})$ and $H(p^{*(k)})$.
		\State Set $e$ to $\mid H(p^{*(k)}) - H(p^{*(k-1)})\mid$.
		\EndWhile\label{euclidendwhile}
		\State $p(y) = \gamma^{(k)}\delta(y) + (1-\gamma^{(k)})\delta^*(y)g^{*(k)}(y)$ is the desired MaxEnt distribution.
	\end{algorithmic}
\end{algorithm}

Estimation of density functions for mixture distributions, which are a mixture of discrete and continuous distributions, using entropy maximization was studied by \citet{Politis1994}. By applying their method to semi-continuous data, it can be shown that $\gamma^*$ in the MaxEnt density function in \eqref{eq:maxent-semicts-density1} is given by 
\begin{equation}
    \label{eq:politis-maxent-gamma-soln}
    \gamma^* = \frac{1}{1+\text{exp}(H(g^*))},
\end{equation}
where $g^*$ is the MaxEnt distribution for the continuous part and $H(g^*)$ is the corresponding maximized entropy value. The trouble with applying this method is that $g^*$ is not known and must be estimated. However, notice that if $\gamma^*$ is known, then $g^*$ can be found by maximizing the entropy subject to the constraints in \eqref{eq:thm1-constr-g^*}. This suggests an alternating maximization algorithm similar to Algorithm~\ref{alg:maxent-semicts-algo} given by 
\begin{algorithm}[H]
	\caption{Alternating Entropy Maximization for Semi-continuous data based on \citet{Politis1994} (AEM-Politis)}
	\label{alg:maxent-semicts-politis-algo}
	\begin{algorithmic}
		\State Step 1a: Choose small positive real values as tolerance error $\epsilon$, and a maximum number of iterations $M$.
		\State Step 1c: Choose an initial value $\gamma^{(0)}$ for $\gamma$ and calculate $H(g^{*(0)})$, where $g^{*(0)}$ is the MaxEnt distribution with constraints in \eqref{eq:thm1-constr1} with $\gamma^{(0)}$ substituted for $\gamma$. Set $e$ to $H(p^{*(0)})$, respectively. Set $k=1$.
		\State Step 2: 
		\While{$k \leq M$ OR $e > \epsilon$}.
		\State Determine $H(g^{*(k)})$ by setting $\gamma$ to $\gamma^{(k-1)}$.
		\State Set $\gamma^{(k)}$ using the equality in \eqref{eq:politis-maxent-gamma-soln}.
		\State Set $e$ to $\mid H(p^{*(k)}) - H(p^{*(k-1)})\mid$.
		\EndWhile\label{euclidendwhile}
		\State $p(y) = \gamma^{(k)}\delta(y) + (1-\gamma^{(k)})\delta^*(y)g^{*(k)}(y)$ is the desired MaxEnt distribution.
	\end{algorithmic}
\end{algorithm}

Hydrological studies on daily rainfall data typically focus only on the positive part of the semi-continuous distribution (\citet{RHO2019210}) and set the probability of observing a zero to the proportion of zeros in the data.   
For example, \citet{Papal2012} have focused on the MaxEnt solution of the continuous part alone and explored various distributions including a generalized gamma distribution, generalized beta distribution etc with applications for the daily rainfall data. More recently, \citet{PopuriZois2019} too set the probability of observing a zero to the sample value and apply the entropy maximization method to the blind source separation problem for semi-continuous data. Notice that while fixing $\gamma^*$ to sample values will simplify the problem, the resulting solution is not the MaxEnt distribution. Furthermore, the real application of the entropy maximization approach is when only the sample values of the constraint functions are known and not the entire data samples to estimate the distribution. As far as we know, there is no method in the literature that can estimate the density function for semi-continuous data using values of the constraint functions alone. As far as we know, Algorithm~\ref{alg:maxent-semicts-algo} is the first such method.

%Here we do not assume that the MaxEnt distributions for the discrete and continuous parts are known.

\section{Numerical Examples}
\label{sec:maxent-sim}
We evaluate Algorithm~\ref{alg:maxent-semicts-algo} using three examples. The first two use simulated data and in the last example, we apply the methods in section~\ref{sec:method} to daily rainfall data in a region in the Midwest US. In these examples, we refer to Algorithm~\ref{alg:maxent-semicts-algo} as the AES method, Algorithm~\ref{alg:maxent-semicts-politis-algo} as the AES-Politis method, and the method where the probability of observing a zero is set to the proportion of zeros in the sample as the `Two-part EM' method.

\begin{example}\label{ex:maxent-exp-sim}(Semi-continuous exponential distribution)
In this example, we evaluate our method in Algorithm \ref{alg:maxent-semicts-algo} on various sets of simulated data from the two-part exponential distribution given by 
\begin{equation}
    \label{eq:ex1-tpexp-distr}
    p(y) = \gamma^*\delta(y) + (1-\gamma^*)\delta^*(y)\lambda\exp(-\lambda y),
\end{equation}
where $\lambda > 0$ is the rate parameter. As shown in Example~\ref{ex:maxent-exp}, if the only constraint is on the first moment, then the MaxEnt solution has a closed form and it belongs to the class of semi-continuous exponential distributions. Therefore, we could exactly compute the entropy value and other properties of the true solution. This helps us to evaluate how well our algorithm is able to recover the entropy value of the true MaxEnt solution. We simulate datasets for a set of choices for parameters $(\gamma, \lambda)$ and calculate the percentage deviation of the entropy and the variance of the estimated MaxEnt solution from the corresponding true values. Specifically, for each choice of $(\gamma, \lambda)$, we simulate data of size $1000$ from \eqref{eq:ex1-tpexp-distr} and estimate the MaxEnt distribution using the three methods: AES, AES-Politis, and Two-part EM. We replicate the simulation $100$ times and compute the percentage deviation values, averaged over the replications. We also compute the convergence path averaged over $100$ replications to study the convergence of our algorithm. Table~\ref{tab:table1} shows the percentage deviations in the estimates of the true $H(p)$ and variance values for the three models for a selected choice of parameter values. Clearly, estimates by AEM are significantly closer to the true values compared to those from the AEM-Politis and the Two-part EM methods. Figure~\ref{fig:exp-conv} shows the convergence behavior of the three methods. These plots show that the AEM method converges fast to the true entropy value. These results show that the AEM method estimates the MaxEnt distribution with the least bias among the three methods.
\begin{table}[H]
\begin{tabular}{SSSSSSSS}
%\toprule
& & \multicolumn{3}{c}{H(p)} & \multicolumn{3}{c}{Variance} \\
\cmidrule(r){3-5}\cmidrule(l){6-8}
$\gamma$ & $\lambda$ & {AEM} & {AEM-Politis} & {Two-part EM}  & {AEM} & {AEM-Politis} & {Two-part EM} \\
\midrule
0.1 & 0.5 & 0.0 & -3.2 & -6.2 & 4.5 & -24.6 & -31.6 \\
0.1 & 1.5 & -0.2 & -6.8 & -33.8 & -1.1 & -33.5 & -56.4 \\
0.4 & 0.5 & -0.1 & -4.5 & -0.5 & -0.3 & -28.5 & 11.6 \\
0.4 & 1.5 & -0.1 & -7.9 & -5.2 & -0.2 & -34.9 & -29.5 \\
0.8 & 0.5 & -0.3 & -8.1 & -22.8 & 0.4 & -34.7 & 171.4 \\
0.8 & 1.5 & 0.0 & -11.1 & -7.4 & 1.5 & -37.1 & 64.1 \\
\bottomrule
\end{tabular}
\caption{\footnotesize Percentage deviation in H(p) and Variance of the estimated MaxEnt distribution $p$ for various parameter choices.}
\label{tab:table1}
\end{table}
\begin{figure}[H]
\begin{tabular}{cc}
  \includegraphics[width=84mm]{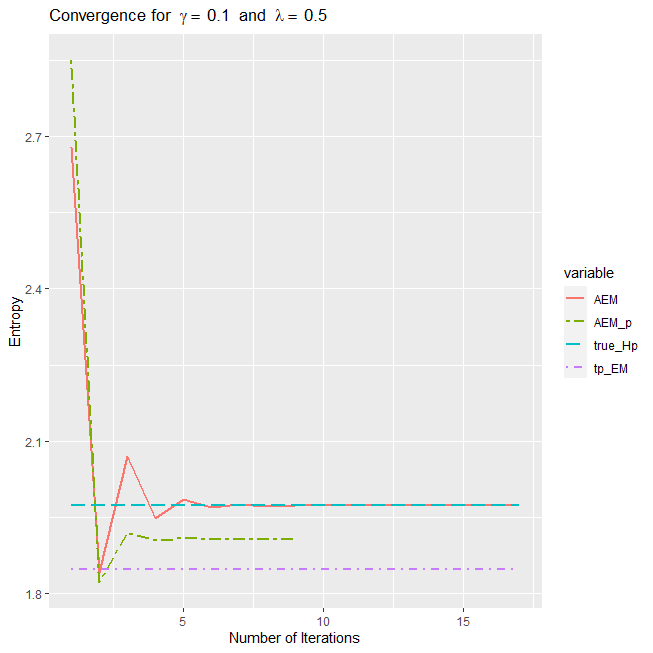} &   \includegraphics[width=84mm]{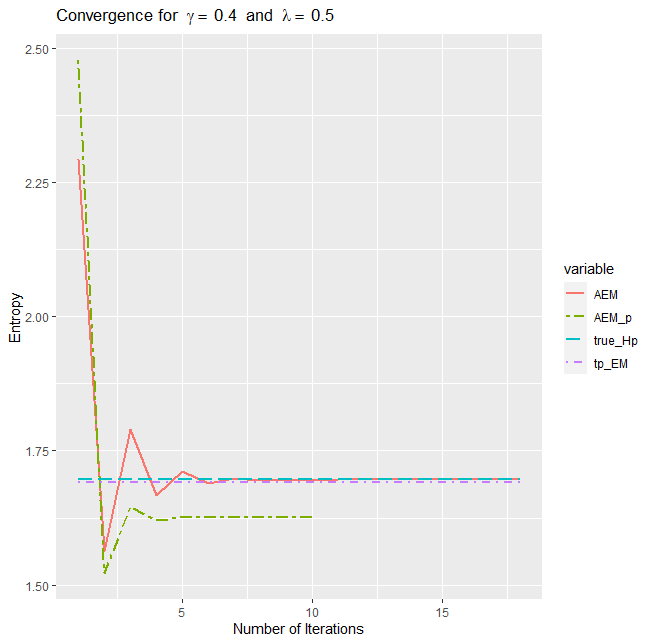} \\
%(a)  & (b) second \\[6pt]
 \includegraphics[width=84mm]{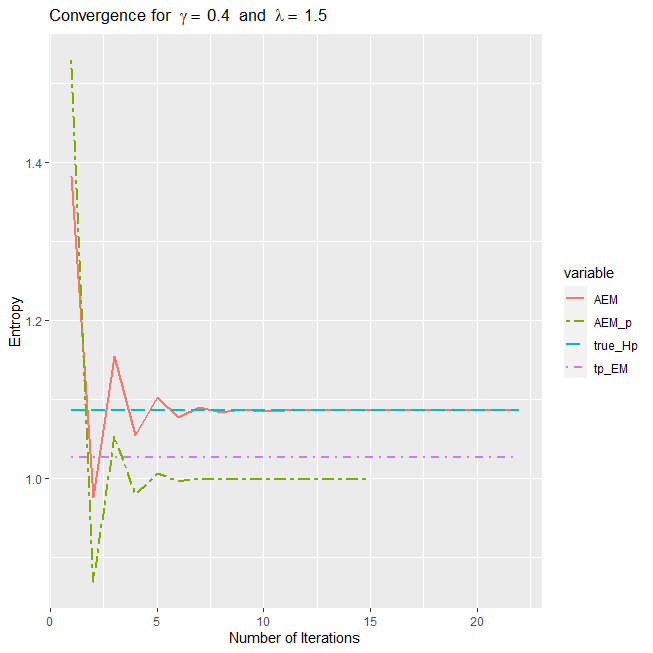} &   \includegraphics[width=84mm]{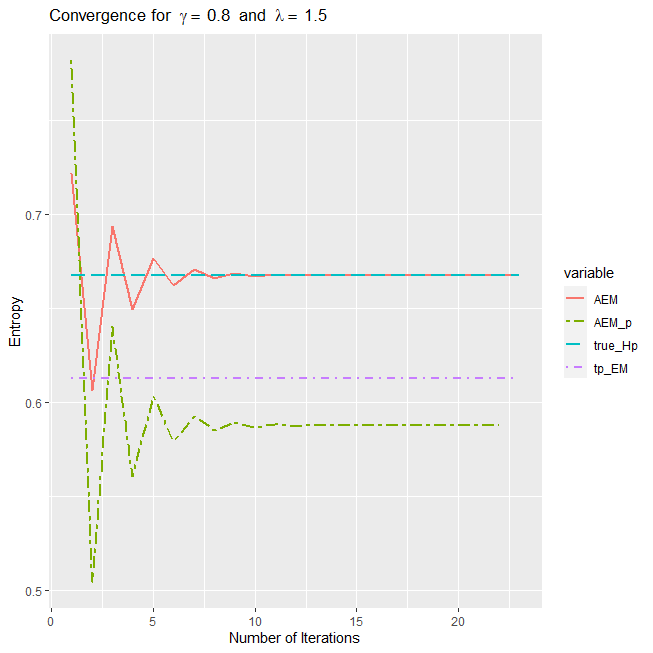} \\
%(c) third & (d) fourth 
\end{tabular}
\caption{Convergence profiles for AEM, AEM-Politis, and Two-part EM methods}
\label{fig:exp-conv}
\end{figure}
\end{example}

\begin{example}\label{ex:maxent-gamma-sim}(Semi-continuous gamma distribution)
In this example we estimate the density function using the three methods: AEM, AEM-Politis, and Two-part EM on datasets simulated from the two-part gamma distribution given by
\begin{equation}
    \label{eq:ex1-tpgamma-distr}
    p(y) = \gamma^*\delta(y) + (1-\gamma^*)\delta^*(y)\frac{1}{\Gamma(\kappa)\theta^\kappa}y^{\kappa-1}e^{-y/\theta}
\end{equation}
for a set of fixed parameters $\gamma$, $\kappa$, and $\theta$. Unlike the semi-continuous exponential distribution in  Example~\ref{ex:maxent-exp-sim}, there is no closed form MaxEnt solution here. Therefore, it is not possible to study the performance of the three algorithms relative to the true entropy value. Instead, we report the maximum entropy values attained by each algorithm. Table~\ref{tab:table2} shows the entropy of the MaxEnt solutions computed by the three algorithms. The entropy values shown are averaged over $100$ replications. For each replication, a sample of size $1000$ was generated from the distribution in \eqref{eq:ex1-tpgamma-distr}. The results show that the solution computed by the AEM algorithm has the maximum entropy value compared to the two other algorithms for the selected parameter values. Our implementations of AEM and AEM-Politis encountered numerical problems while solving the nonlinear equations in Example~\ref{ex:maxent-gamma} for larger values of $\gamma$ and therefore, we do not report more results but we expect the behavior of the estimated entropy values to be similar. 
%\begin{comment}
\begin{table}[H]
\centering
\begin{tabular}{SSSSSS}
$\gamma$ & $\kappa$ & $\theta$ & {AEM} & {AEM-Politis} & {Two-part EM} \\
\midrule
0.1 & 1.0 & 0.5 & 0.89 & 0.82 & 0.65 \\
0.1 & 3.0 & 0.5 & 1.76 & 1.62 & 1.36 \\
0.1 & 3.0 & 1.0 & 2.08 & 2.01 & 1.98 \\
0.1 & 5.0 & 0.5 & 1.85 & 1.76 & 1.63 \\
0.4 & 1.0 & 1.5 & 1.51 & 1.38 & 1.50 \\
0.4 & 5.0 & 1.0 & 2.17 & 2.13 & 1.96 \\
0.4 & 5.0 & 1.5 & 2.63 & 2.59 & 2.20 \\
\bottomrule
\end{tabular}
\caption{\footnotesize Entropy (averaged over replications) of the estimated MaxEnt distribution $p$ for various parameter choices.}
\label{tab:table2}
\end{table}
%\end{comment}
\end{example}

\begin{example}\label{ex:pr-example}(Daily Precipitation data)\newline
In this example, we use the observed daily rainfall data from about $500$ locations in a region in the Missouri river basin (MRB), one of the largest river basins in the US. Figure~\ref{fig:prmrb} shows the study region, which spreads across several states in the Midwest US. The daily rainfall data shown in Figure~\ref{fig:obshist} is from a particular location in that region. Since the true data generating processes are not known, we do not know the true entropy value. In hydrological studies, daily precipitation data is typically analyzed using a two-part gamma distribution (\citet{Papal2012}). Therefore, we use constraints on the mean and logarithm of the positive data in the entropy maximization problem and apply the AEM method and compare it with the AEM-Politis and Two-part EM methods. We selected around $500$ locations in the study region where the proportion of zeros in the observed data is less than $0.6$, in order to avoid numerical problems. At each of these locations, the three algorithms are run to obtain three entropy values. We take the percentage difference between the entropy value from the AEM method and the maximum of entropy values from the AEM-Politis and the Two-part EM methods. In our analysis, entropy values estimated by the AEM method were greater than the maximum of entropy values estimated by the AEM-Politis and Two-part EM methods by $1\%$ to $10\%$ with at around $125$ out of $500$ of the locations the difference was greater than $5\%$. This clearly shows that our method was able to maximize the entropy for the two-part gamma distribution consistently better than the rest of the methods when applied to the observed daily rainfall datasets.
\begin{figure}[H]
 	\centering
 	\includegraphics[width=0.6\textwidth]{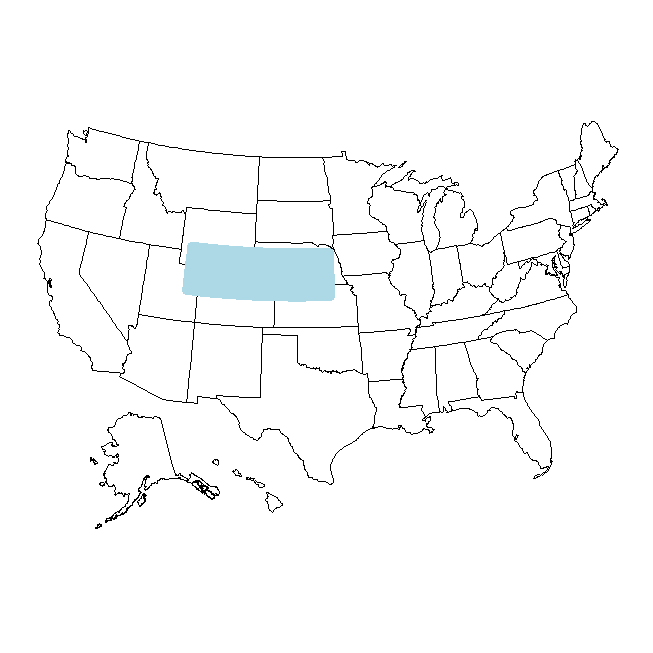}
 	\caption{Region of the test data}
 	\label{fig:prmrb}
 \end{figure}
\end{example}

\section{Conclusion}
\label{sec:discussion}
In this paper, we have presented a new algorithm to estimate the density function for semi-continuous data using entropy maximization. The novelty in our method lies in the fact that it is applied to the entire distribution of semi-continuous data without assuming that the probability of observed a zero is known. The algorithm only needs the values of the constraint functions from the observed data and not the entire data, unlike current methods in the literature. We have also provided a theoretical justification for our algorithm and presented a comprehensive development of entropy maximization theory for semi-continuous data. Using simulations, we have demonstrated that our method converges to the true entropy value and the existing methods in literature produce biased solutions. A useful application of our algorithm is in approximating posterior predictive distributions, for example, in Bayesian weather generators that need predictive distributions as input for climate variables instead of point estimates. We would like to explore this practical application in future. We would also like to conduct a comprehensive assessment of how well the distributions estimated by our method model the observed daily rainfall patterns.

\section*{Acknowledgments}
The hardware used in the computational studies is part of the
UMBC High Performance Computing Facility (HPCF).
The facility is supported by the U.S. National Science Foundation
through the MRI program
(grant nos. CNS-0821258, CNS-1228778, and OAC-1726023)
and the SCREMS program (grant no. DMS-0821311),
with additional substantial support from the
University of Maryland, Baltimore County (UMBC).
See hpcf.umbc.edu
for more information on HPCF and the projects using its resources.

\bibliographystyle{plainnat}
\bibliography{maxentreferences}
\end{document}